\documentclass[12pt]{article}

\usepackage[margin=1in]{geometry}
\usepackage{amsthm}
\usepackage{amssymb}
\usepackage{amsmath}
\usepackage{graphicx}
\usepackage{url}
\usepackage{color}
\usepackage{framed}
\usepackage[table]{xcolor}
\usepackage[ruled,vlined]{algorithm2e}
\usepackage{rotating}

\usepackage{braket} % for \Set
\usepackage{xparse,mathtools}

\newtheorem{theorem}{Theorem}%[section]

\newtheorem{lemma}[theorem]{Lemma}

\theoremstyle{definition}
\newtheorem{problem}[theorem]{Problem}
\newtheorem{example}[theorem]{Example}
\newtheorem{definition}[theorem]{Definition}
\DeclareMathOperator{\opkernel}{Ker}
\DeclareMathOperator{\opimage}{Im}
\DeclareMathOperator{\opsgn}{sgn}

\DeclareMathOperator{\opspan}{span}
\DeclareMathOperator{\opreal}{Re}
\DeclareMathOperator{\opim}{Im}
\DeclareMathOperator{\opexp}{exp}
\DeclareMathOperator{\oplog}{log}
\DeclareMathOperator{\optrace}{trace}

\DeclarePairedDelimiter\parens{\lparen}{\rparen}

\NewDocumentCommand\prob{d()}{\mathbb{P}\IfValueT{#1}{\parens*{#1}}}
\NewDocumentCommand\expec{d()}{\mathbb{E}\IfValueT{#1}{\parens*{#1}}}

\NewDocumentCommand\vspan{d()}{\opspan\IfValueT{#1}{\parens*{#1}}}
\NewDocumentCommand\kernel{d()}{\opkernel\IfValueT{#1}{\parens*{#1}}}
\NewDocumentCommand\image{d()}{\opimage\IfValueT{#1}{\parens*{#1}}}
\NewDocumentCommand\sgn{d()}{\opsgn\IfValueT{#1}{\parens*{#1}}}
\NewDocumentCommand\re{d()}{\opreal\IfValueT{#1}{\parens*{#1}}}
\NewDocumentCommand\im{d()}{\opim\IfValueT{#1}{\parens*{#1}}}

\NewDocumentCommand\expf{d()}{\opexp\IfValueT{#1}{\parens*{#1}}}
\NewDocumentCommand\logf{d()}{\oplog\IfValueT{#1}{\parens*{#1}}}
\NewDocumentCommand\trace{d()}{\optrace\IfValueT{#1}{\parens*{#1}}}

\title{Parameter estimation in the SIR model\\from early infections}

\author{Charles~Clum\footnote{Department of Mathematics, The Ohio State University, Columbus, OH} \and Dustin~G.~Mixon\footnote{Department of Mathematics, The Ohio State University, Columbus, OH}}
\date{}

\begin{document}
\maketitle

\begin{abstract}
A standard model for epidemics is the SIR model on a graph.
We introduce a simple algorithm that uses the early infection times from a sample path of the SIR model to estimate the parameters this model, and we provide a performance guarantee in the setting of locally tree-like graphs.
\end{abstract}

\section{Introduction}

During an epidemic, government leaders are expected to help maintain public health while simultaneously preventing an economic meltdown.
In the absence of a vaccine, decision makers must choose between various non-pharmaceutical interventions.
This decision requires an informative forecast of the epidemic at a very early time.
To obtain such a forecast, it is helpful to have a parametrized model for epidemics.
What follows is a particularly popular compartmental model that originates from the classic work of Kermack and McKendrick \cite{KermackM:27}.

\begin{definition}[SIR model]
Fix a simple, connected graph $G$ and parameters $\lambda,\mu\geq0$.
Consider a continuous-time Markov chain in which the state is a partition $(S,I,R)$ of $V(G)$.
For the initial state, draw $v\sim\mathsf{Unif}(V(G))$ and put
\[
S(0)=V(G)\setminus\{v\},
\qquad
I(0)=\{v\},
\qquad
R(0)=\emptyset.
\]
Given the current state $(S,I,R)$, then for every $u\in S$ that is adjacent to some member of $I$, the process transitions
\[
(S,I,R)
\to
(S\setminus\{u\},I\cup\{u\},R)
\]
with rate $\lambda |N(u)\cap I|$, while for each $w\in I$, the process transitions
\[
(S,I,R)
\to
(S,I\setminus\{w\},R\cup\{w\})
\]
with rate $\mu$.
\end{definition}

In the real world, it is difficult to distinguish between vertices in the infected set $I(t)$ and vertices in the recovered set $R(t)$ at any time $t$.
For example, Li et al.~\cite{Li:20} estimated the early transmission dynamics of COVID-19 in Wuhan, China by collecting infection times and identifying exposures through contact tracing.
In this paper, we model this lack of information by assuming it is known when a vertex is infected, but unknown when an infected vertex recovers.
We let $U(t)=(S(t))^c=I(t)\cup R(t)$ denote the random set of unsusceptible vertices at time $t$.

\begin{problem}
Given $G$ and $\{U(t)\}_{t\in[0,\epsilon]}$ for some small $\epsilon>0$, estimate $\lambda$ and $\mu$.
\end{problem}

Unfortunately, it is not always possible to estimate $\lambda$ and $\mu$ when $\epsilon$ is small.
This can be seen with a popular instance of the SIR model in which $G$ is the complete graph:

\begin{example}
Suppose $G=K_n$.
By symmetry, it suffices to consider the cardinalities
\[
s(t):=|S(t)|,
\qquad
i(t):=|I(t)|,
\qquad
r(t)=|R(t)|.
\]
In fact, $(s(t),i(t),r(t))$ is also a continuous-time Markov chain in this case.
The initial conditions are $s(0)=n-1$, $i(0)=1$, $r(0)=0$, and the process transitions
\[
(s,i,r)
\to
(s-1,i+1,r)
\]
with rate $\lambda is$ and
\[
(s,i,r)
\to
(s,i-1,r+1)
\]
with rate $\mu i$.
Assuming $n$ is large and $\lambda n=:\beta$, then putting $\sigma:=s/n$, $\iota:=i/n$, and $\rho:=r/n$, we may pass to the mean-field approximation:
\[
\frac{d\sigma}{dt}=-\beta \iota\sigma,
\qquad
\frac{d\iota}{dt}=\beta \iota\sigma-\mu \iota,
\qquad
\frac{d\rho}{dt}=\mu \iota.
\]
This approximation is popular because it is much easier to interact with.
The approximation is good once the number of infected vertices becomes a fraction of $n$, and the approximation is better when this fraction is larger~\cite{Kurtz:71}.
This suggests an initial condition of the form 
\[
\sigma(t_0)=1-\delta-\gamma,
\qquad
\iota(t_0)=\delta,
\qquad
\rho(t_0)=\gamma
\]
for some small $t_0,\delta,\gamma>0$.
For simplicity, we translate time so that $t_0=0$.

We argue there is no hope of determining $(\beta,\mu)$ from data of the form $\{\iota(t)+\rho(t)\}_{t\in[0,\epsilon]}$ for small $\epsilon>0$.
(While the following argument is not rigorous, it conveys the main idea.)
Notice that for $t\in[0,\epsilon]$, it holds that $\sigma(t)\approx 1$, and so $\iota(t)\approx\delta e^{(\beta-\mu)t}$ and
\[
\rho(t)
=\gamma+\mu\int_0^t\iota(s)ds
\approx\gamma+\frac{\mu}{\beta-\mu}\delta e^{(\beta-\mu)s}\Big|_0^t
=\gamma+\frac{\mu}{\beta-\mu}\delta\Big(e^{(\beta-\mu)t}-1\Big).
\]
Then our data takes the form
\begin{align*}
\iota(t)+\rho(t)
&\approx \delta e^{(\beta-\mu)t}+\gamma+\frac{\mu}{\beta-\mu}\delta\Big(e^{(\beta-\mu)t}-1\Big)\\
&=\Big(\gamma-\frac{\mu}{\beta-\mu}\delta\Big)+\frac{\beta}{\beta-\mu}\delta\cdot e^{(\beta-\mu)t}
=:a+be^{ct}.
\end{align*}
We can expect to determine $a$, $b$, and $c$ by curve fitting.
However, we don't know $\delta$ or $\gamma$, but rather their sum.
As such, we claim that $(a,b,c)$ only determines $\beta-\mu$.
Indeed, for every choice of $(\beta,\mu)$ such that $\beta-\mu=c$, it could be the case that
\[
\delta=\frac{c}{\beta}\cdot b,
\qquad
\gamma=a+\frac{\mu}{\beta}\cdot b,
\]
which would then be consistent with the data $(a,b,c)$.
Of course, additional information about $(\beta,\mu)$ could conceivably be extracted from higher-order terms, since $a+be^{ct}$ is merely an approximation of $\iota(t)+\rho(t)$.
However, we expect any such signal to be dwarfed by noise in the data.

While the short-term behavior of $\iota$ is exponential with rate $\beta-\mu$, the long-term behavior is instead governed by the quotient $R_0:=\beta/\mu$, known as the basic reproductive number.
This can be seen by dilating time by substituting $s=\mu t$.
In this variable, the mean-field approximation instead takes the form
\[
\frac{d\sigma}{ds}
=\frac{d\sigma}{dt}\cdot\frac{dt}{ds}
=-R_0 \iota\sigma,
\qquad
\frac{d\iota}{ds}
=\frac{d\iota}{dt}\cdot\frac{dt}{ds}
=R_0 \iota\sigma- \iota,
\qquad
\frac{d\rho}{ds}
=\frac{d\rho}{dt}\cdot\frac{dt}{ds}
= \iota.
\]
That is, $R_0$ (together with initial conditions) determines $\iota$ modulo time dilation, and notably, whether the curve ever exceeds the capacity of the medical care system.
However, $R_0=\beta/\mu$ cannot be determined from $\beta-\mu$.
\hfill$\bigtriangleup$
\end{example}

Overall, the complete graph is not amenable to determining $(\lambda,\mu)$ from $\{U(t)\}_{t\in[0,\epsilon]}$.
However, real-world social networks are far from complete.
Like social networks, expander graphs have low degree, but considering their spectral properties, one might presume that they are just as opaque as the complete graph.
Surprisingly, this is not correct!
In this paper, we show how certain graphs (including certain expander graphs) are provably amenable to determining $(\lambda,\mu)$ from $\{U(t)\}_{t\in[0,\epsilon]}$.

In the following section, we introduce our approach.
Specifically, we isolate infections that pass across bridges in a local subgraph of the social network, and then we estimate $\lambda$ and $\mu$ from these infection statistics.
Section~3 gives the proof of our main result: that our approach provides decent estimates of $\lambda$ and $\mu$ in the setting of locally tree-like graphs.
Our proof makes use of the vast literature on SIR dynamics on infinite trees.
We conclude in Section~4 with a discussion of opportunities for future work.

\section{Parameter estimation from controlled infections}

We start with the simple example in which $G=K_2$.
According to the SIR process, one of the two vertices is infected, and then it either infects the other vertex or it recovers before doing so.
Let $Z$ denote the random amount of time it takes for the second vertex to become infected.
Notice that $Z=\infty$ with probability $\frac{\mu}{\lambda+\mu}$.
On the other hand, if we condition on the event $Z<\infty$, then the distribution of $Z$ is exponential with rate $\lambda+\mu$.
(This is a consequence of the fact that the minimum and minimizer of independent exponential random variables are independent.)
Notice that if we could estimate
\[
\mathbb{P}\{Z<\infty\}
=\frac{\lambda}{\lambda+\mu},
\qquad
\mathbb{E}[Z|Z<\infty]
=\frac{1}{\lambda+\mu},
\]
then we could recover $(\lambda,\mu)$, as desired.
For example, if we had access to multiple independent draws of the SIR model on $K_2$, then we could obtain such estimates.
For certain types of graphs, we can actually simulate this setup, and this is the main idea of our approach.

In practice, we will not have the time to determine whether $Z<\infty$, and so we instead truncate $Z\leftarrow\min\{Z,\tau\}$ for some threshold $\tau>0$.
In particular, we write $Z\sim\mathsf{CI}(\lambda,\mu,\tau)$ to denote a random variable with distribution
\[
\mathbb{P}\{Z\in(a,b)\}
=\int_{\max(a,0)}^{\min(b,\tau)} \lambda e^{-(\lambda+\mu)t}dt
+\frac{\lambda+\mu e^{-(\lambda+\mu)\tau}}{\lambda+\mu}\cdot 1_{\tau\in(a,b)}.
\]
We seek to estimate $\lambda$ and $\mu$ given $\tau$ and estimates of the following quantities:
\begin{align*}
p
&:=\mathbb{P}\{Z<\tau\}
=\frac{\lambda}{\lambda+\mu}\cdot(1-e^{-(\lambda+\mu)\tau}),\\
q
&:=\mathbb{E}[Z|Z<\tau]
=\frac{1}{\lambda+\mu}\cdot\frac{1-((\lambda+\mu)\tau+1)e^{-(\lambda+\mu)\tau}}{1-e^{-(\lambda+\mu)\tau}}.
\end{align*}
First, we show that good estimates of $p$ and $q$ yield good estimates of $\lambda$ and $\mu$:

\begin{lemma}
\label{eq.lambda and mu from p and q}
Suppose $m:=(\lambda+\mu)\tau\geq 2$, and take $P$ and $Q$ such that
\[
e^{-\epsilon}
\leq \frac{P}{p}
\leq e^{\epsilon},
\qquad
e^{-\epsilon}
\leq \frac{1-P}{1-p}
\leq e^{\epsilon},
\qquad
e^{-\epsilon}
\leq \frac{Q}{q}
\leq e^{\epsilon}
\]
for some $\epsilon>0$.
Then
\begin{align*}
e^{-2\epsilon}\lambda
&\leq\frac{P}{Q}
\leq e^{2\epsilon}(1+2(m+1)e^{-m})\cdot \lambda,\\
e^{-2\epsilon}\mu
&\leq \frac{1-P}{Q}
\leq e^{2\epsilon}(1+2(\lambda/\mu+m+1)e^{-m})\cdot \mu.
\end{align*}
\end{lemma}

\begin{proof}
First, observe that
\begin{align*}
\frac{p}{q}
&=\lambda\cdot \frac{(1-e^{-m})^2}{1-(m+1)e^{-m}}
=:\lambda\cdot \alpha_1,\\
\frac{1-p}{q}
&=\mu\cdot\frac{(1+(\lambda/\mu)e^{-m})(1-e^{-m})}{1-(m+1)e^{-m}}
=:\mu\cdot \alpha_2.
\end{align*}
Since $m\geq 2$, we have $(m+1)e^{-m}\leq 1/2$, and so
\[
1
\leq \alpha_1
\leq \frac{1}{1-(m+1)e^{-m}}
\leq 1+2(m+1)e^{-m},
\]
and
\begin{align*}
1
\leq \alpha_2
\leq \frac{1+(\lambda/\mu)e^{-m}}{1-(m+1)e^{-m}}
&\leq (1+(\lambda/\mu)e^{-m})(1+2(m+1)e^{-m})\\
&=1+(\lambda/\mu+2(m+1))e^{-m}+2(\lambda/\mu)(m+1)e^{-2m}\\
&\leq 1+2(\lambda/\mu+m+1)e^{-m}.
\end{align*}
Thus,
\[
e^{-2\epsilon}\lambda
\leq e^{-2\epsilon}\alpha_1\lambda
=e^{-2\epsilon}\cdot\frac{p}{q}
\leq\frac{P}{Q}
\leq e^{2\epsilon}\cdot\frac{p}{q}
=e^{2\epsilon}\alpha_1\lambda
\leq e^{2\epsilon}(1+2(m+1)e^{-m})\cdot \lambda,
\]
and similarly for $(1-P)/Q$.
\end{proof}

Next, we produce estimates $P$ and $Q$ given independent realizations of $Z$:

\begin{lemma}
\label{lem.estimate by averages}
Given independent realizations $\{Z_{a}\}_{a\in A}$ of $Z\sim\mathsf{CI}(\lambda,\mu,\tau)$, put
\[
L_a:=1_{\{Z_{a}<\tau\}},
\qquad
A':=\{a\in A:Z_{a}<\tau\},
\qquad
P:=\frac{1}{|A|}\sum_{a\in A}L_a,
\qquad
Q:=\frac{1}{|A'|}\sum_{a\in A'}Z_a.
\]
Select $\epsilon>0$.
Then with probability $\geq 1-6e^{-c_\epsilon |A| p(1-p)}$, it holds that
\[
e^{-\epsilon}
\leq \frac{P}{p}
\leq e^{\epsilon},
\qquad
e^{-\epsilon}
\leq \frac{1-P}{1-p}
\leq e^{\epsilon},
\qquad
e^{-\epsilon}
\leq \frac{Q}{q}
\leq e^{\epsilon}.
\]
\end{lemma}

\begin{proof}
For convenience, we put $k:=|A|$ and identify $A=[k]$.
We have $\mathbb{E}L_i=p$, $\operatorname{Var}L_i=p(1-p)$, and $L_i\in[0,1]$ almost surely.
As such, we may apply Bernstein's inequality for bounded random variables (see Theorem~2.8.4 in~\cite{Vershynin:18}):
\begin{align*}
\mathbb{P}\Big\{|P-p|\geq \delta p(1-p)\Big\}
&=\mathbb{P}\bigg\{\Big|\sum_{i=1}^k\big(L_i-\mathbb{E}L_i\big)\Big|\geq k\delta\min(p,1-p)\bigg\}\\
&\leq 2\operatorname{exp}\bigg(-\frac{(k\delta p(1-p))^2/2}{kp(1-p)+(k\delta p(1-p))/3}\bigg)\\
&=2\operatorname{exp}\bigg(-\frac{\delta^2/2}{1+\delta/3}\cdot k\cdot p(1-p)\bigg)\\
&\leq2\operatorname{exp}\bigg(-\frac{1}{4}\cdot\min(\delta^2,\delta/3)\cdot k\cdot p(1-p)\bigg).
\end{align*}
Put $\delta=1-e^{-\epsilon}$.
Then both of the following hold with probability $1-2e^{-c_\epsilon k p(1-p)}$:
\begin{align*}
\Big|\frac{P}{p}-1\Big|
&=\frac{|P-p|}{p}
\leq\frac{|P-p|}{p(1-p)}
\leq 1-e^{-\epsilon},\\
\Big|\frac{1-P}{1-p}-1\Big|
&=\frac{|P-p|}{1-p}
\leq\frac{|P-p|}{p(1-p)}
\leq 1-e^{-\epsilon}.
\end{align*}
Note that this implies 
\[
e^{-\epsilon}
\leq \frac{P}{p}
\leq e^{\epsilon},
\qquad
e^{-\epsilon}
\leq \frac{1-P}{1-p}
\leq e^{\epsilon}.
\]
Next, we estimate $Q$.
Conditioned on $A'$, the random variables $\{Z_{a}\}_{a\in A'}$ are all distributed like a $\tau$-truncated version $Y$ of a random variable $X\sim\mathsf{Exp}(\lambda+\mu)$, and there exist absolute constants $C_1,C_2>0$ such that
\[
\|Y-\mathbb{E}Y\|_{\psi_1}
\leq C_1\|Y\|_{\psi_1}
\leq C_1\|X\|_{\psi_1}
= \frac{C_1C_2}{\lambda+\mu}.
\]
As such, we may apply Bernstein's inequality for subexponential random variables (see Theorem~2.8.1 in~\cite{Vershynin:18}):
\begin{align*}
\mathbb{P}\Big\{|Q-q|\geq (1-e^{-\epsilon})q\Big\}
&\leq \mathbb{P}\Big\{|Q-q|\geq\frac{1-e^{-\epsilon}}{\lambda+\gamma}\Big\}\\
&\leq \mathbb{P}\Big\{|Q-q|\geq\frac{1-e^{-\epsilon}}{\lambda+\gamma}\Big| |A'|\geq e^{-\epsilon}pk\Big\}+\mathbb{P}\Big\{|A'|<e^{-\epsilon}pk\Big\}\\
&\leq \mathbb{P}\Big\{|Q-q|\geq\frac{1-e^{-\epsilon}}{\lambda+\gamma}\Big| |A'|\geq e^{-\epsilon}pk\Big\}+\mathbb{P}\Big\{P<e^{-\epsilon}p\Big\}\\
&\leq 2e^{-c_\epsilon kp}+2e^{-c_\epsilon k p(1-p)}.
\end{align*}
As such, with probability $\geq 1-4e^{-c_\epsilon k p(1-p)}$, it holds that
\[
e^{-\epsilon}
\leq \frac{Q}{q}
\leq e^{\epsilon}.
\]
The result follows from the union bound.
\end{proof}

If we had access to the infected vertices at time $t_0$, we could use the formulas in Lemma~\ref{lem.estimate by averages} to obtain estimators of the SIR parameters that provide a good approximation to the true parameters:

\begin{lemma}
\label{lem.key lemma}
Consider the SIR model on a graph G with parameters $\lambda$ and $\mu$.
Select $r,t_0,t_1>0$ and put $\tau:=t_1-t_0$.
Let $B(r)$ denote the subgraph of $G$ induced by vertices of distance at most $r$ from $U(0)$.
The set of bridges in $B(r)$ with one vertex in $I(t_0)$ and another vertex in $S(t_0)$ takes the form $\{\{a,b\}:a\in A_I,b\in B_a\}$, where $A_I\subseteq I(t_0)$.
For each $a\in A_I$, independently draw $b(a)\sim\mathsf{Unif}(B_a)$.
Let $T(v)$ denote the infection time of $v\in V(G)$, where we take $T(v):=\infty$ if $v$ is never infected.
Consider the random variables
\[
Z_{a}:=\min(T(b(a))-t_0,\tau),
\qquad
L_a:=1_{\{Z_{a}<\tau\}},
\qquad
A_I':=\{a\in A_I:Z_{a}<\tau\},
\]
\[
P_I:=\frac{1}{|A_I|}\sum_{a\in A_I}L_a,
\qquad
Q_I:=\frac{1}{|A_I'|}\sum_{a\in A_I'}Z_a,
\qquad
\hat\lambda_I:=\frac{P_I}{Q_I},
\qquad
\hat\mu_I:=\frac{1-P_I}{Q_I}.
\]
For any fixed $k,\epsilon>0$, define the events
\[
\mathcal{E}_1
:=\{U(t_1)\cap\partial U(t_1)\subseteq V(B(r))\},
\qquad
\mathcal{E}_2
:=\{|\overline{A}|\geq k\},
\]
\[
\mathcal{F}
:=\{\hat\lambda_I\not\in[e^{-\delta}\lambda,e^{\delta}\lambda]\}\cup\{\hat\mu_I\not\in[e^{-\delta}\mu,e^{\delta}\mu]\},
\]
where $\delta:=2\epsilon+2(\lambda/\mu+(\lambda+\mu)\tau+1)e^{-(\lambda+\mu)\tau}$.
Then $\mathbb{P}(\mathcal{F}\cap\mathcal{E}_1\cap\mathcal{E}_2)\leq 6e^{-c_\epsilon kp(1-p)}$, where $p:=\mathbb{P}\{Z<\tau\}$ with $Z\sim\mathsf{CI}(\lambda,\mu,\tau)$.
\end{lemma}

\begin{proof}
Let $S$ denote the vertices in $G$ of distance exactly $r$ from $U(0)$.
Notice that for every vertex $v\in V(B(r))\setminus S$, the edges incident to $v$ in $B(r)$ are precisely the edges incident to $v$ in $G$.
As such, the SIR processes on $G$ and on $B(r)$ are identical until the stopping time
\[
T:=\inf\{t:|U(t)\cap S|>0\}.
\]
For each quantity $T(v),Z_{a},A_I,A_I',B_a,b(a),P,Q,\hat\lambda_I,\hat\mu_I$ defined in the statement of the lemma, there is a corresponding quantity defined by replacing the SIR process on $G$ with the SIR process on $B(r)$, and we denote these variables by $\tilde{T}(v),\tilde{Z}_{a},\tilde{A},\tilde{A}',\tilde{B}_a,\tilde{b}(a),\tilde{P},\tilde{Q},\tilde\lambda,\tilde\mu$.
Each of these variables equals its counterpart over the event $\mathcal{E}_1=\{T>t_1\}$.
In fact, taking $\tilde{\mathcal{F}}$ to similarly correspond to $\mathcal{F}$, then $\tilde{\mathcal{F}}\cap\mathcal{E}_1=\mathcal{F}\cap\mathcal{E}_1$.
This implies
\[
\mathbb{P}(\mathcal{F}\cap\mathcal{E}_1\cap\mathcal{E}_2)
=\mathbb{P}(\tilde{\mathcal{F}}\cap\mathcal{E}_1\cap\mathcal{E}_2)
\leq\mathbb{P}(\tilde{\mathcal{F}}\cap\mathcal{E}_2)
\leq\mathbb{P}(\tilde{\mathcal{F}}|\mathcal{E}_2).
\]
It remains to bound $\mathbb{P}(\tilde{\mathcal{F}}|\mathcal{E}_2)$.

Conditioned on $\tilde{A}$ and $\{\tilde{b}(a)\}_{a\in\tilde{A}}$, then for each $a\in\tilde{A}$, the Markov property implies that $\tilde{Z}_{a}$ has distribution $\mathsf{CI}(\lambda,\mu,\tau)$.
Also, the vertices $\{\tilde{b}(a)\}_{a\in \tilde{A}}$ are pairwise distinct almost surely.
Indeed, if $\tilde{b}(a)=\tilde{b}(a')$, then if we delete the edge $\{a,\tilde{b}(a)\}$, we can still traverse a walk from $a$ to $U(0)$ to $a'$ to $\tilde{b}(a')=\tilde{b}(a)$, implying $\{a,\tilde{b}(a)\}$ was not a bridge.
As such, conditioned on $\tilde{A}$ and $\{\tilde{b}(a)\}_{a\in\tilde{A}}$, the variables $\{\tilde{Z}_{a}\}_{a\in\tilde{A}}$ are jointly independent.
Then Lemmas~\ref{eq.lambda and mu from p and q} and~\ref{lem.estimate by averages} together imply 
\[
\mathbb{P}(\tilde{\mathcal{F}}|\mathcal{E}_2)
=\mathbb{E}[\mathbb{P}(\tilde{\mathcal{F}}|\tilde{A},\{\tilde{b}(a)\}_{a\in\tilde{A}})|\mathcal{E}_2]
\leq\mathbb{E}[6e^{-c_\epsilon |\tilde{A}|p(1-p)}|\mathcal{E}_2]
\leq6e^{-c_\epsilon kp(1-p)}.
\qedhere
\]
\end{proof}

\begin{algorithm}[t]
\SetAlgoLined
\KwData{Graph $G=(V,E)$, parameters $r,t_0,t_1>0$, infection times $T\colon V\to\mathbb{R}\cup\{\infty\}$}
\KwResult{Estimate $(\hat\lambda,\hat\mu)$ of SIR parameters}
Let $B(r)$ denote the subgraph of $G$ induced by vertices of distance at most $r$ from the minimizer of $T$.\\
Put $U:=\{u\in V(B(r)):T(u)\leq t_0\}$.\\
Let $\{\{a,b\}:a\in A,b\in B_a\}$ denote the bridges in $B(r)$ with one vertex in $U\supseteq A$ and the other vertex in $V(B(r))\setminus U\supseteq\bigcup_{a\in A}B_a$.\\
For each $a\in A$, independently draw $b(a)\sim\mathsf{Unif}(B_a)$.\\
Put $\tau:=t_1-t_0$.\\
For each $a\in A$, put $Z_{a}:=\min(T(b(a))-t_0,\tau)$.\\
For each $a\in A$, put $L_a:=1_{\{Z_{a}<\tau\}}$.\\
Put $A':=\{a\in A:Z_a<\tau\}$.\\
Put $P:=\frac{1}{|A|}\sum_{a\in A}L_a$ and $Q:=\frac{1}{|A'|}\sum_{a\in A'}Z_{a}$.\\
Output $\hat\lambda:=\frac{P}{Q}$ and $\hat\mu:=\frac{1-P}{Q}$.
 \caption{SIR parameter estimation from controlled infections\label{alg.main_algorithm}}
\end{algorithm}

Of course, in our setup, we do not have access to $I(t_0)$, but rather $U(t_0)=I(t_0)\cup R(t_0)$, and so we cannot apply Lemma~\ref{lem.key lemma} directly.
Instead, we assume that $\lambda$ is sufficiently large compared to $\mu$ that $U(t_0)$ is a decent approximation of $I(t_0)$.
This approach is summarized in Algorithm~\ref{alg.main_algorithm}.
As we will see, the approximation $I(t_0)\approx U(t_0)$ introduces some error in our estimators.
To analyze the performance of Algorithm~\ref{alg.main_algorithm}, it is convenient to focus on a certain (large) family of graphs.
We say a graph $G$ is \textbf{$(r,\eta)$-locally tree-like} if for a fraction $1-\eta$ of the vertices $v\in V(G)$, it holds that the subgraph induced by the vertices of distance at most $r$ from $v$ is a tree.
For example, it is known that for every fixed choice of $d\in\mathbb{N}$ with $d>1$ and $c\in(0,\frac{1}{4})$, there exists $\gamma>0$ such that a random $d$-regular graph on $n$ vertices is $(c\log_{d-1}n,n^{-\gamma})$-locally tree-like with probability approaching $1$ as $n\to\infty$; see Proposition~4.1 in~\cite{BauerschmidtHY:19}.
By focusing on this class of graphs, we may apply the vast literature on SIR dynamics on infinite trees to help analyze the performance of Algorithm~\ref{alg.main_algorithm}.

\begin{theorem}[Main Result]
\label{thm.parameter estimation locally tree like}
Fix $d\in\mathbb{N}$ with $d\geq 3$ and $c,\gamma>0$, and consider any sequence of $d$-regular, $(c\log_{d-1}n,n^{-\gamma})$-locally tree-like graphs on $n$ vertices with $n\to\infty$.
Select any $\alpha$ and $\beta$ such that $0<\alpha<\beta<\tfrac{c}{2e(d-1)\lambda}$, and put
\[
r:=c\log_{d-1}n,
\qquad
t_0:=\alpha\log_{d-1}n,
\qquad
t_1:=\beta\log_{d-1}n.
\]
For each $n$, let $\mathcal{E}_0$ denote the event that the subgraph $B(r)$ induced by vertices of distance at most $r$ from $U(0)$ is a tree, and let $\mathcal{E}_\infty$ denote the event that $U(\infty)$ contains a vertex of distance greater than $r$ from $U(0)$.
Suppose $\lambda\geq6\mu>0$.
Then one of the following holds:
\begin{itemize}
\item[(a)]
$\limsup_{n\to\infty}\mathbb{P}((\mathcal{E}_0\cap\mathcal{E}_\infty)^c)=1$, meaning there is a subsequence of $n$ for which, with probability at least $1-o(1)$, no vertex outside of $B(r)$ will ever be infected, in which case the infection does not spread to even a constant fraction of the graph.
\item[(b)]
For each $n$, the following holds:
Conditioned on $\mathcal{E}_0\cap\mathcal{E}_\infty$, Algorithm~\ref{alg.main_algorithm} returns
\begin{equation}
\label{eq.success}
\hat\lambda\in[\tfrac{1}{8}\lambda-o(1),\lambda+o(1)],
\qquad
\hat\mu\in[\mu-o(1),8\mu+o(1)]
\end{equation}
with probability tending to $1$ as $n\rightarrow\infty$.
\end{itemize}
Furthermore, $\mathbb{P}(\mathcal{E}_0\cap\mathcal{E}_\infty)\geq 1-\frac{\mu}{(d-2)\lambda-\mu}-o(1)$.
\end{theorem}

\begin{figure}[t]
\begin{center}
\begin{turn}{90}
\hspace{1.1in}
\begin{turn}{270}
\footnotesize{$\lambda$}
\end{turn}
\end{turn}
\hspace{-0.1in}
\includegraphics[width=0.4\textwidth,trim={110 220 110 220},clip]{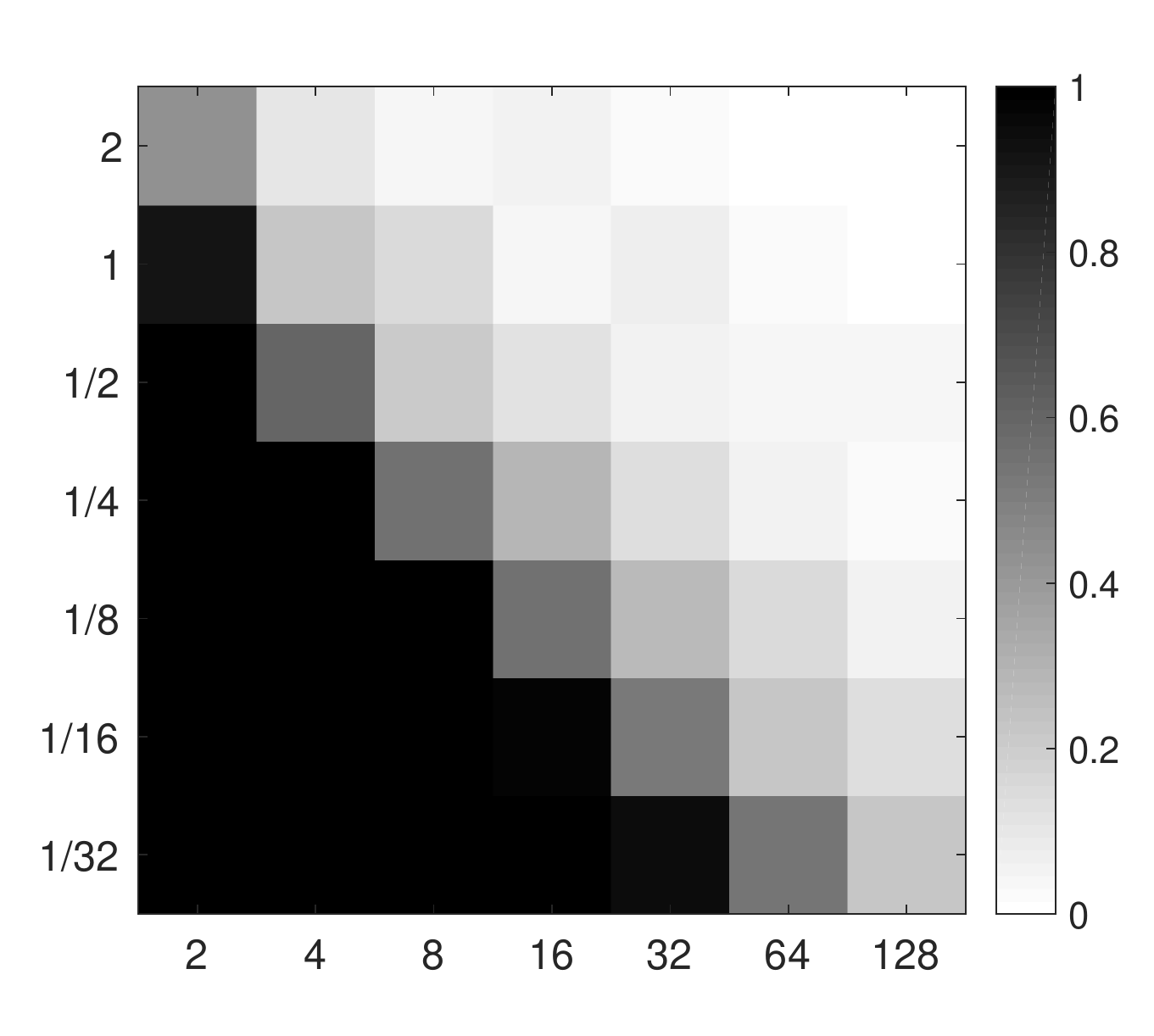}
\hspace{0.5in}
\begin{turn}{90}
\hspace{1.1in}
\begin{turn}{270}
\footnotesize{$\lambda$}
\end{turn}
\end{turn}
\hspace{-0.1in}
\includegraphics[width=0.4\textwidth,trim={110 220 110 220},clip]{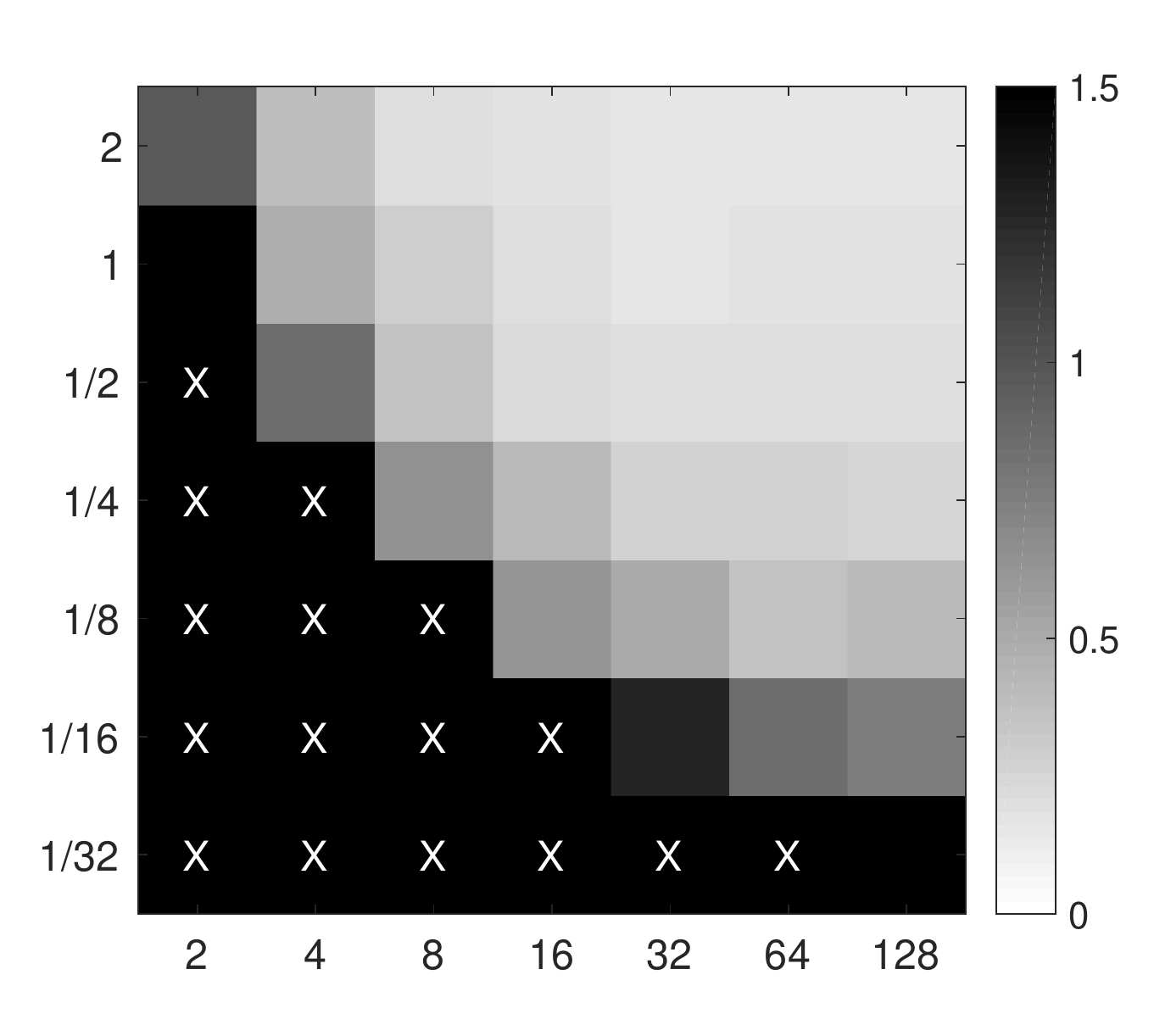}\\
\vspace{-0.1in}
\footnotesize{$d$}\hspace{3.2in}\footnotesize{$d$}
\end{center}
\caption{\label{fig.phase_transitions}
For each $d\in\{2^k:k\in\{1,\ldots,7\}\}$, draw $100$ independent sample paths of the SIR process initialized at the root of the infinite $d$-ary tree for each $\lambda\in\{2^j:j\in\{-5,\ldots,1\}\}$ with $\mu=1$.
Consider $T_0:=\inf\{t:|U(t)|\geq 100\}$.
\textbf{(left)}
Proportion of sample paths with $T_0=\infty$.
For these instances, the lack of spread renders parameter estimation moot.
This plot mimics the phase transition in Lemma~\ref{lem.bound third term}.
\textbf{(right)}
Run Algorithm~\ref{alg.main_algorithm} with $r=\infty$, $t_0=T_0$, and $t_1=T_0+4$.
If the algorithm breaks for at least $80$ of the trials (either because $T_0=\infty$ or $Q=0$), plot a black square with a white ``X.''
Otherwise, plot the average of $\max((|\hat\lambda-\lambda|)/\lambda,(|\hat\mu-\mu|)/\mu)$ over the trials for which the algorithm does not break.
Apparently, Algorithm~\ref{alg.main_algorithm} performs better than predicted by Theorem~\ref{thm.parameter estimation locally tree like}.
}
\end{figure}

The proof is given in the next section, and Figure~\ref{fig.phase_transitions} illustrates the actual behavior of Algorithm~\ref{alg.main_algorithm} for comparison.
The factors of $8$ in \eqref{eq.success} are due to the approximation $I(t_0)\approx U(t_0)$, and they have not been optimized.
We suspect that these factors can be replaced by terms that approach $1$ as $d\to\infty$, but this requires a different technique.
We also suspect the hypothesis $\lambda\geq 6\mu$ is an artifact of our proof.
As the following lemma indicates, the threshold $(d-2)\lambda>\mu$ would be more natural; this threshold arises from standard results on Galton--Watson processes (see the lecture notes~\cite{Alsmeyer:online}, for example).

\begin{lemma}
\label{lem.bound third term}
Consider any sequence $\{G_n\}$ of $d$-regular graphs on $n$ vertices with $n\to\infty$ such that $G_n$ is $(r_n,\eta_n)$-locally tree-like with $r_n\to\infty$ and $\eta_n\to0$.
For each $n$, consider the SIR model on $G_n$ with parameters $\lambda$ and $\mu$, and let $\mathcal{E}_\infty$ denote the event that $U(\infty)$ contains a vertex of distance greater than $r$ from $U(0)$.
\begin{itemize}
\item[(a)]
If $(d-2)\lambda\leq\mu$, then $U(\infty)\subseteq B(r_n)$ with probability at least $1-o(1)$ as $n\to\infty$.
\item[(b)]
If $(d-2)\lambda>\mu$, then for each $n$, it holds that $\mathbb{P}(\mathcal{E}_\infty^c)\leq \frac{\mu}{(d-2)\lambda-\mu}+\eta_n$.
\end{itemize}
\end{lemma}

\begin{proof}
For (a), it suffices to show $\mathbb{P}(\{U(\infty)\not\subseteq B(r_n)\}|\mathcal{E}_0)=o(1)$, since then
\begin{align*}
\mathbb{P}\{U(\infty)\not\subseteq B(r_n)\}
&=\mathbb{P}(\{U(\infty)\not\subseteq B(r_n)\}\cap\mathcal{E}_0^c)+\mathbb{P}(\{U(\infty)\not\subseteq B(r_n)\}\cap\mathcal{E}_0)\\
&\leq\mathbb{P}(\mathcal{E}_0^c)+\mathbb{P}(\{U(\infty)\not\subseteq B(r_n)\}|\mathcal{E}_0)\cdot\mathbb{P}(\mathcal{E}_0)\\
&=o(1).
\end{align*}
Restrict to the event $\mathcal{E}_0$, and put $\{v\}=U(0)$.
Deleting $v$ from $B(r_n)$ produces $d$ connected components, each of which can be viewed as a subgraph of the infinite $(d-1)$-ary tree $T$ that is rooted by the corresponding neighbor $w$ of $v$.
The SIR evolution on $T$ determines a Galton--Watson process that gives the eventual number $Z_m$ of unsusceptible vertices at distance $m\geq1$ from $w$:
\[
Z_m
=\sum_{k=1}^{Z_m-1}X_{m,k},
\]
where $X_{m,k}$ denotes the number of vertices infected by the $k$th infected vertex in $T$ that has distance $m$ from $w$.
The random variables $\{X_{n,k}\}_{n,k\geq1}$ are independent with distribution matching a random variable that we denote $X$.

We can describe the distribution of $X$ as follows.
Draw random variables $T\sim \mathsf{Exp}(\mu)$ and $C_{1},\ldots, C_{d-1}\sim \mathsf{Exp}(\lambda)$.
If $N$ denotes the number of $k\in [d-1]$ for which $C_{k}<T$, then $N$ is distributed like $X$.
Conditioned on $T$, this number is a binomial with parameters $d-1$ and $1 - e^{-\lambda T}$.
Hence,
\begin{align*}
\mathbb{E}X 
= \mathbb{E}[\mathbb{E}[N|T]] 
= (d-1)\mathbb{E}[1 - e^{-\lambda T}]
= (d-1)\cdot\tfrac{\lambda}{\lambda + \mu}
\leq 1.
\end{align*}
Since in addition, it holds that $\mathbb{P}\{X>1\}>0$, Theorem~1.7 in~\cite{Alsmeyer:online} gives that $\sum_{m=1}^\infty Z_m$ is finite almost surely.
Put $M:=\sup\{m:Z_m>0\}$.
Then $M<\infty$ almost surely.
In particular, a union bound over the $d$ different neighbors of $v$ gives
\[
\mathbb{P}(\{U(\infty)\not\subseteq B(r_n)\}|\mathcal{E}_0)
\leq d\cdot\mathbb{P}\{M>r_n\}
=o(1),
\]
where the last step uses the fact that the survivor function of $M$ vanishes at infinity.

For (b), it suffices to show $\mathbb{P}(\mathcal{E}_\infty^c|\mathcal{E}_0)\leq \frac{\mu}{(d-2)\lambda-\mu}$, since then
\begin{align*}
\mathbb{P}(\mathcal{E}_\infty^c)
=\mathbb{P}(\mathcal{E}_\infty^c \cap\mathcal{E}_0^c)+\mathbb{P}(\mathcal{E}_\infty^c \cap\mathcal{E}_0)
\leq\mathbb{P}(\mathcal{E}_0^c)+\mathbb{P}(\mathcal{E}_\infty^c|\mathcal{E}_0)\mathbb{P}(\mathcal{E}_0)
\leq \eta_n+\tfrac{\mu}{(d-2)\lambda-\mu}.
\end{align*}
Restrict to the event $\mathcal{E}_0$, and put $\{v\}=U(0)$.
As before, we identify a Galton--Watson process to analyze, but this one is slightly different:
Delete one of the neighbors of $v$ from $B(r_n)$ and identify the connected component $C$ containing $v$ with a subgraph of a $(d-1)$-ary tree $T$ with root $v$.
The SIR evolution on $T$ determines a Galton--Watson process that gives the eventual number $Z_m$ of unsusceptible vertices at distance $m\geq1$ from $v$:
\[
Z_m
=\sum_{k=1}^{Z_m-1}X_{m,k},
\]
where $X_{m,k}$ denotes the number of vertices infected by the $k$th infected vertex in $T$ that has distance $m$ from $v$.
The random variables $\{X_{m,k}\}_{m,k\geq1}$ are independent with distribution matching a random variable that we denote $X$.
We see that $\mathbb{P}(\mathcal{E}_\infty^c|\mathcal{E}_0)$ is at most the extinction probability $q$ of this process.
By Theorem~1.7 in~\cite{Alsmeyer:online}, $q$ is the smallest fixed point of the generating function $f(s):=\mathbb{E}[s^X]$ for $s\in[0,1]$.

Put $p_0:=\mathbb{P}\{X=0\}$ and $p_1:=\mathbb{P}\{X=1\}$.
Since $f(0)=p_0>0$, then for every $s>0$ that satisfies $f(s)\leq s$, we have $q\leq s$ by the intermediate value theorem.
Since
\[
f(\tfrac{p_0}{1-p_0-p_1})
\leq p_0+p_1(\tfrac{p_0}{1-p_0-p_1})+(1-p_0-p_1)(\tfrac{p_0}{1-p_0-p_1})^2
=\tfrac{p_0}{1-p_0-p_1},
\]
it follows that $q\leq \tfrac{p_0}{1-p_0-p_1}$.
Before estimating $p_0$ and $p_1$, it is helpful to introduce some notation.
An infected parent with $d-1$ children infects $X$ of these children.
The parent recovers exponentially with rate $\mu$, and we let $R$ denote the recovery time.
Simultaneously, each child is infected exponentially with rate $\lambda$, and so we denote independent random variables $I_i\sim\mathsf{Exp}(\lambda)$ such that 
\[
\left\{\begin{array}{cl}
I_i&\text{if }I_i<R\\
\infty&\text{otherwise}
\end{array}\right.
\]
gives the time of infection for child $i$.
Then
\begin{align*}
p_0
=\mathbb{P}\{X=0\}
=\mathbb{E}[\mathbb{P}(\{X=0\}|R)]
&=\mathbb{E}[1-\mathbb{P}(\{X>0\}|R)]\\
&=1-\mathbb{E}[\mathbb{P}(\{\min_i I_i<R|R)]\\
&=1-\mathbb{E}[1-e^{-(d-1)\lambda R}]
=\mathbb{E}[e^{-(d-1)\lambda R}]
=\tfrac{\mu}{\mu+(d-1)\lambda}.
\end{align*}
Next, the order statistic $I_{(2)}$ has the same distribution as $\tfrac{1}{(d-1)\lambda}Z_1+\tfrac{1}{(d-2)\lambda}Z_2$, where $Z_1,Z_2\sim\mathsf{Exp}(1)$ are independent.
It follows that
\begin{align*}
p_0+p_1
=\mathbb{P}\{X\leq 1\}
&=\mathbb{E}[\mathbb{P}(\{X\leq 1\}|R)]\\
&=1-\mathbb{E}[\mathbb{P}(\{I_{(2)}<R\}|R)]\\
&=(d-1)(d-2)\lambda\cdot\mathbb{E}[\tfrac{1}{(d-2)\lambda}e^{-(d-2)\lambda R}-\tfrac{1}{(d-1)\lambda}e^{-(d-1)\lambda R}]\\
&=(d-1)\cdot\tfrac{\mu}{\mu+(d-2)\lambda}-(d-2)\cdot\tfrac{\mu}{\mu+(d-1)\lambda}.
\end{align*}
This identity combined with the fact that $x\mapsto\frac{\mu x}{\mu+\lambda x}$ is an increasing function gives
\[
p_0+p_1
\leq\Big((d-1)-(d-2)\Big)\Big(\tfrac{\mu}{\mu+(d-1)\lambda}+\tfrac{\mu}{\mu+(d-2)\lambda}\Big)
=\tfrac{\mu}{\mu+(d-1)\lambda}+\tfrac{\mu}{\mu+(d-2)\lambda}
\leq \tfrac{2\mu}{\mu+(d-2)\lambda}.
\]
Overall, since $(d-2)\lambda>\mu$, we have
\[
q
\leq \tfrac{p_0}{1-p_0-p_1}
\leq (1-\tfrac{2\mu}{\mu+(d-2)\lambda})^{-1}\cdot\tfrac{\mu}{\mu+(d-1)\lambda}
\leq (1-\tfrac{2\mu}{\mu+(d-2)\lambda})^{-1}\cdot\tfrac{\mu}{\mu+(d-2)\lambda}
=\tfrac{\mu}{(d-2)\lambda-\mu}.
\qedhere
\]
\end{proof}

\section{Proof of Theorem~\ref{thm.parameter estimation locally tree like}}

The last statement follows from Lemma~\ref{lem.bound third term}.
To prove the remainder of the result, we will assume that (a) does not hold and prove that (b) holds.
Since (a) does not hold, there exists some $\kappa > 0$ such that $\mathbb{P}(\mathcal{E}_{0}\cap \mathcal{E}_{\infty}) > \kappa$ for all sufficiently large $n$.
Select $n$ sufficiently large, and let $\mathcal{F}$ denote the failure event that \eqref{eq.success} does not hold for some $o(1)$ function to be identified later.
We wish to show that $\mathbb{P}(\mathcal{F}|\mathcal{E}_{0}\cap \mathcal{E}_{\infty}) = o(1)$.
Let $\mathcal{E}_1$ denote the event that $U(t_1)\cup\partial U(t_1)$ is contained within distance $r$ of $U(0)$.
In particular, on the event $\mathcal{E}_{0}\cap \mathcal{E}_{1}$, all of infected and recovered vertices at time $t_{1}$ reside in a tree with root $U(0)$.   
Also, selecting $k:=\frac{\alpha(\mu+\lambda)}{4}\cdot r=\frac{\alpha(\mu+\lambda)}{4}\cdot c\log_{d-1}n$, let $\mathcal{E}_{2}$ denote the event that $|U(t_0)|\geq k$.

We will make use of two simple inequalities involving arbitrary events $A$, $B$, and $C$.
First,
\begin{equation}
\label{simple bound 1}
\mathbb{P}(A|B)
=\tfrac{\mathbb{P}(A\cap B\cap C)}{\mathbb{P}(B)}+\tfrac{\mathbb{P}(A\cap B\cap C^c)}{\mathbb{P}(B)}
\leq\tfrac{\mathbb{P}(A\cap B\cap C)}{\mathbb{P}(B\cap C)}+\tfrac{\mathbb{P}(B\cap C^c)}{\mathbb{P}(B)}
=\mathbb{P}(A|B\cap C)+\mathbb{P}(C^c|B).
\end{equation}
Furthermore, if $\mathbb{P}(C^{c}|B)\neq 1$, then
\begin{equation}
\label{simple bound 2}
\mathbb{P}(A|B\cap C)
=\frac{\mathbb{P}(A\cap B\cap C)}{\mathbb{P}(C|B)\mathbb{P}(B)}
\leq\frac{\mathbb{P}(A\cap B)}{\mathbb{P}(C|B)\mathbb{P}(B)}
=\frac{\mathbb{P}(A|B)}{1 - \mathbb{P}(C^{c}|B)}.
\end{equation}
Two applications of \eqref{simple bound 1} gives  
\begin{align}
\mathbb{P}(\mathcal{F}|\mathcal{E}_0\cap\mathcal{E}_\infty)
\nonumber
&\leq\mathbb{P}(\mathcal{F}|\mathcal{E}_0\cap\mathcal{E}_1\cap\mathcal{E}_\infty)+\mathbb{P}(\mathcal{E}_1^c|\mathcal{E}_0\cap\mathcal{E}_\infty)\\
\label{eq.terms to bound}
&\leq\mathbb{P}(\mathcal{F}|\mathcal{E}_0\cap\mathcal{E}_1\cap\mathcal{E}_2\cap\mathcal{E}_\infty)
+\mathbb{P}(\mathcal{E}_2^c|\mathcal{E}_0\cap\mathcal{E}_1\cap\mathcal{E}_\infty)
+\mathbb{P}(\mathcal{E}_1^c|\mathcal{E}_0\cap\mathcal{E}_\infty).
\end{align}
To bound the third term in \eqref{eq.terms to bound}, we let $\mathbb{P}_{\lambda,\mu}$ denote the probability measure corresponding to the SIR model with parameters $\lambda$ and $\mu$.
Then, for sufficiently large $n$, we have
\begin{equation}
\label{eq.bound on third term}
\mathbb{P}(\mathcal{E}_1^c|\mathcal{E}_0\cap\mathcal{E}_\infty)
=\frac{\mathbb{P}_{\lambda,\mu}(\mathcal{E}_0\cap\mathcal{E}_1^c\cap\mathcal{E}_\infty)}{\mathbb{P}_{\lambda,\mu}(\mathcal{E}_0\cap\mathcal{E}_\infty)}
\leq\frac{\mathbb{P}_{\lambda,0}(\mathcal{E}_0\cap\mathcal{E}_1^c\cap\mathcal{E}_\infty)}{\mathbb{P}_{\lambda,\mu}(\mathcal{E}_0\cap\mathcal{E}_\infty)}
\leq\frac{\mathbb{P}_{\lambda,0}(\mathcal{E}_0\cap\mathcal{E}_1^c)}{\kappa}.
\end{equation}
As such, it suffices to show that $\mathbb{P}_{\lambda,0}(\mathcal{E}_0\cap\mathcal{E}_1^c)=o(1)$.
We accomplish this by analyzing the corresponding branching process:

\begin{lemma}
\label{lem.branching process}
Let $G$ denote the infinite $d$-ary tree.
Consider the process $\{H_t\}_{t\geq0}$ of induced subgraphs of $G$ in which $H_0$ is induced by the root vertex, and then for each $v\in V(G)\setminus V(H_t)$ that is a $G$-child of some vertex in $H_t$, it holds that $v$ is added to $H_t$ at unit rate.
Let $B_m$ denote the first time at which a vertex of distance $m$ from the root vertex of $G$ is added to $H_t$.
Then $B_m\geq m/(2ed)$ with probability $\geq1-e^{-m/2}$.
\end{lemma}

\begin{proof}
Let $B_{m,r}$ denote the $r$th time at which a vertex of distance $m$ from the root vertex of $G$ is added to $H_t$, and note that $B_m=B_{m,1}$.
Let $c,\theta>0$ be given (to be selected later).
Then Markov's inequality gives
\begin{align*}
\mathbb{P}\{B_m\leq cm\}
=\mathbb{P}\{e^{-\theta B_m}\geq e^{-\theta cm}\}
&\leq e^{\theta cm}\cdot\mathbb{E}e^{-\theta B_m}\\
&\leq e^{\theta cm}\cdot\mathbb{E}\bigg[\sum_{r=1}^de^{-\theta B_{m,r}}\bigg]
= e^{\theta cm}\cdot\Bigg(\mathbb{E}\bigg[\sum_{r=1}^de^{-\theta B_{1,r}}
\bigg]\Bigg)^m,
\end{align*}
where the last identity, which appears in Theorem~1 in~\cite{Kurtz:71}, follows from analyzing a certain martingale.
In our case, $X_0:=B_{1,1}\sim\mathsf{Exp}(d)$ and $X_r:=B_{1,r+1}-B_{1,r}\sim\mathsf{Exp}(d-r)$ for $r\in\{1,\ldots, d-1\}$ are independent.
It follows that
\[
\mathbb{E}e^{-\theta B_{1,r}}
=\mathbb{E}e^{-\theta \sum_{k=0}^{r-1}X_k}
=\prod_{k=0}^{r-1}\mathbb{E}e^{-\theta X_k}
=\prod_{k=0}^{r-1}\frac{d-k}{d-k+\theta}
\leq\Big(\frac{d}{d+\theta}\Big)^r,
\]
and so
\[
\mathbb{E}\bigg[\sum_{r=1}^de^{-\theta B_{1,r}}
\bigg]
\leq \sum_{r=1}^d\Big(\frac{d}{d+\theta}\Big)^r
\leq \frac{1}{1-\frac{d}{d+\theta}}-1
=\frac{d}{\theta}.
\]
Overall,
\[
\mathbb{P}\{B_m\leq cm\}
\leq e^{\theta cm}\cdot(d/\theta)^m
=e^{-m(\log(\theta/d)-c\theta)}.
\]
Selecting $\theta=ed$ and $c=1/(2ed)$ gives the result.
\end{proof}

\begin{lemma}
\label{lem.bound on e1c given e0}
Suppose $G$ is $(r,\eta)$-locally tree-like, consider the SIR process on $G$ with $\mu=0$, and take any $t_1\leq \tfrac{r-2}{2e(d-1)\lambda}$.
Then $\mathbb{P}(\mathcal{E}_1^c|\mathcal{E}_0)\leq 2e^{-(r-2)/2}$.
\end{lemma}

\begin{proof}
By time dilation, we may take $\lambda=1$ without loss of generality.
Condition on $\mathcal{E}_0$.
After time $T\sim\mathsf{Exp}(d)$, there are two infected vertices $u$ and $v$.
Removing the edge $\{u,v\}$ from $H$ produces two $(d-1)$-ary trees, with root vertices $u$ and $v$.
Extend these trees to infinite $(d-1)$-ary trees $H_u$ and $H_v$.
Let $B_u$ denote the first time at which a vertex of distance $r-2$ from the root vertex is infected in $H_u$, and similarly for $B_v$.
Then by assumption on $t_1$, we have
\begin{align*}
\mathcal{E}_1^c\cap\mathcal{E}_0
&\subseteq
\Big(\{T+B_u<t_1\}\cup\{T+B_v<t_1\}\Big)\cap\mathcal{E}_0\\
&\subseteq
\Big(\{B_u<\tfrac{r-2}{2e(d-1)}\}\cup\{B_v<\tfrac{r-2}{2e(d-1)}\}\Big)\cap\mathcal{E}_0.
\end{align*}
Finally, we apply the union bound and Lemma~\ref{lem.branching process} to get
\[
\mathbb{P}(\mathcal{E}_1^c|\mathcal{E}_0)
\leq\mathbb{P}(\{B_u<\tfrac{r-2}{2e(d-1)}\}|\mathcal{E}_0)+\mathbb{P}(\{B_v<\tfrac{r-2}{2e(d-1)}\}|\mathcal{E}_0)
\leq 2e^{-(r-2)/2}.
\qedhere
\]
\end{proof}

Overall, \eqref{eq.bound on third term} and Lemma~\ref{lem.bound on e1c given e0} together give $\mathbb{P}(\mathcal{E}_{1}^{c}|\mathcal{E}_{0}\cap \mathcal{E}_{\infty})=o(1)$.
Next, we may combine this bound with \eqref{simple bound 2} to bound the second term in \eqref{eq.terms to bound}:
\begin{align}
\label{eq.bound second term 1}
\mathbb{P}(\mathcal{E}_2^c|\mathcal{E}_0\cap\mathcal{E}_1\cap\mathcal{E}_\infty)\le \frac{\mathbb{P}(\mathcal{E}_{2}^{c}|\mathcal{E}_{0}\cap \mathcal{E}_{\infty})}{1 - \mathbb{P}(\mathcal{E}_{1}^{c}|\mathcal{E}_{0}\cap \mathcal{E}_{\infty})} = \frac{\mathbb{P}(\mathcal{E}_{2}^{c}|\mathcal{E}_{0}\cap \mathcal{E}_{\infty})}{1 - o(1)} \le \frac{\mathbb{P}(\mathcal{E}_{2}^{c}\cap \mathcal{E}_{0}\cap \mathcal{E}_{\infty})}{(1- o(1))\kappa}.
\end{align}
Considering $\mathbb{P}(\mathcal{E}_{2}^{c}\cap \mathcal{E}_{0}\cap \mathcal{E}_{\infty})\leq\mathbb{P}(\mathcal{E}_2^c\cap\mathcal{E}_\infty)\leq\mathbb{P}(\mathcal{E}_2^c|\mathcal{E}_\infty)$, it suffices to show that $\mathbb{P}(\mathcal{E}_2^c|\mathcal{E}_\infty)=o(1)$.
To this end, it is convenient to consider the stopping time $T_0:=\inf\{t:|U(t)|\geq k\}$, noting that $\{T_{0} > t_{0}\}=\mathcal{E}_{2}^{c}$.  
Since $k = \frac{\alpha(\mu+\lambda)}{4}r = \frac{c(\mu+\lambda)}{4}t_{0}$ and $c<1$, it follows that
\begin{equation}
\label{eq.bound second term1}
\mathbb{P}(\mathcal{E}_2^c | \mathcal{E}_{\infty})
= \mathbb{P}(\{T_{0} > t_{0}\}|\mathcal{E}_{\infty})
= \mathbb{P}(\{T_{0} > \tfrac{4k}{c(\mu+\lambda)}\}|\mathcal{E}_{\infty})
\leq \mathbb{P}(\{T_{0}\geq \tfrac{4k}{\mu+\lambda}\}|\mathcal{E}_{\infty}).
\end{equation}
Next, our assumptions that $\alpha<\frac{c}{2e(d-1)\lambda}$, $(d-2)\lambda>\mu$, and $c<1$ together imply
\[
\tfrac{\alpha(\mu+\lambda)}{4}
<\tfrac{c}{2e(d-1)\lambda}\cdot\tfrac{\mu+\lambda}{4}
<\tfrac{c}{8e}
<1,
\]
from which it follows that $k<r$.
On the event $\mathcal{E}_{\infty}$, it holds that $U(\infty)$ induces a tree that contains a path of length greater than $r$, from which is follows that $|U(\infty)|>r$.
As such, $\mathcal{E}_{\infty}\subseteq \{|U(\infty)|> r\}\subseteq \{|U(\infty)|\geq k\}= \{T_{0} < \infty\}$.
This allows us to continue \eqref{eq.bound second term1}:
\begin{align}
\mathbb{P}(\mathcal{E}_2^c | \mathcal{E}_{\infty})
\nonumber
&\leq\frac{\mathbb{P}(\{T_{0}\geq \tfrac{4k}{\mu+\lambda}\}\cap\mathcal{E}_{\infty})}{\mathbb{P}(\mathcal{E}_{\infty})}\\
\label{eq.bound on second term2}
&\leq\frac{\mathbb{P}(\{T_{0}\geq \tfrac{4k}{\mu+\lambda}\}\cap\{T_{0} < \infty\})}{\mathbb{P}(\mathcal{E}_0\cap\mathcal{E}_{\infty})}
\leq\frac{1}{\kappa}\cdot\mathbb{P}(\{T_{0}\geq \tfrac{4k}{\mu+\lambda}\}|\{T_{0} < \infty\}).
\end{align}
To continue, we show that $\mathbb{P}(\{T_{0}\geq \tfrac{4k}{\mu+\lambda}\}|\{T_{0} < \infty\})=o(1)$:

\begin{lemma}
\label{lem.bound first term}
Put $T_{0} := \inf\{t\geq 0:|U(t)| \geq k\}$.
Then
\[
\mathbb{P}(\{T_0> \tfrac{4k}{\mu+\lambda}\}|\{T_0<\infty\})
\leq e^{-c_1 k},
\]
where $c_1>0$ is a universal constant.
\end{lemma}

\begin{proof}
The random number $N$ of transitions that occur over the interval $(0,T_0]$ is given by $N=|U(T_0)|+|R(T_0)|-1\le 2|U(T_{0})|-1\le 2k$.
Conditioned on the sequence of states of the discrete time Markov chain $\{M_n\}_{n\geq0}$, the transition times $\{X_n\}_{n\geq1}$ are independent and exponentially distributed with (deterministic) parameters $\lambda e(I_n,S_n)+\mu|I_n|$.
For any sequence of states in the event $\{T_0<\infty\}$, the first $N$ of these parameters are all at least $\mu+\lambda$.
Put $\mathcal{E}:=\{T_0<\infty\}$ and denote $\mathbb{P}_\mathcal{E}(A):=\mathbb{P}(A|\mathcal{E})$.
Then
\[
\mathbb{P}_\mathcal{E}\{T_0> t_0 \}
=\mathbb{P}_\mathcal{E}\bigg\{\sum_{n=1}^N X_n> t_0 \bigg\}\\
\leq\mathbb{P}_\mathcal{E}\bigg\{\sum_{n=1}^N (X_n-\mathbb{E}X_n)> t_0-\frac{2k}{\mu+\lambda}\bigg\}.
\]
Put $t_0:=\frac{4k}{\mu+\lambda}$.
Conditioning on $N$, we may apply Bernstein's inequality for subexponential random variables (see Theorem~2.8.1 in~\cite{Vershynin:18}).
In particular, we let $C>0$ denote a universal constant such that any random variable of the form $X\sim\mathsf{Exp}(\lambda)$ satisfies $\|X-\mathbb{E}X\|_{\psi_1}\leq\frac{C}{\lambda}$.
Then
\begin{align*}
\mathbb{P}_\mathcal{E}\{T_0> t_0 \}
&\leq\mathbb{P}_\mathcal{E}\bigg\{\Big|\sum_{n=1}^N (X_n-\mathbb{E}X_n)\Big|> t_0-\frac{2k}{\mu+\lambda} \bigg\}\\
&\leq\mathbb{E}_\mathcal{E}\bigg[\operatorname{exp}\bigg(-c\min\Big(\frac{4k^{2}}{N C^2},\frac{2k}{C}\Big)\bigg)\bigg]
\leq e^{-c_1 k},
\end{align*}
where the last step applies the bound $N\leq 2k$ and $c_1:=\frac{2c}{\max(C^2,C)}$.
\end{proof}

Overall, \eqref{eq.bound second term 1}, \eqref{eq.bound on second term2} and Lemma~\ref{lem.bound first term} together give that the second term in \eqref{eq.terms to bound} is $o(1)$.
It remains to show that the first term in \eqref{eq.terms to bound} is $o(1)$.
To this end, we first show that $\mathbb{P}(\mathcal{E}_{0}\cap \mathcal{E}_{1}\cap \mathcal{E}_{2}\cap \mathcal{E}_{\infty})\geq(1-o(1))\kappa$.
Notice that
\[
\mathbb{P}(\mathcal{E}_2^c|\mathcal{E}_0\cap\mathcal{E}_\infty)
=\frac{\mathbb{P}(\mathcal{E}_2^c\cap\mathcal{E}_0\cap\mathcal{E}_\infty)}{\mathbb{P}(\mathcal{E}_0\cap\mathcal{E}_\infty)}
\leq\frac{\mathbb{P}(\mathcal{E}_2^c|\mathcal{E}_\infty)}{\kappa}
=o(1),
\]
where the last step applies \eqref{eq.bound on second term2} and Lemma~\ref{lem.bound first term}.
Combining this with \eqref{eq.bound on third term} and Lemma~\ref{lem.bound on e1c given e0} after a union bound then gives
\begin{align*}
\mathbb{P}(\mathcal{E}_0\cap\mathcal{E}_1\cap\mathcal{E}_2\cap\mathcal{E}_\infty)
&=\mathbb{P}(\mathcal{E}_1\cap\mathcal{E}_2|\mathcal{E}_0\cap\mathcal{E}_\infty)\cdot\mathbb{P}(\mathcal{E}_0\cap\mathcal{E}_\infty)\\
&\geq\kappa\cdot\mathbb{P}(\mathcal{E}_1\cap\mathcal{E}_2|\mathcal{E}_0\cap\mathcal{E}_\infty)\\
&\geq\kappa\cdot(1-\mathbb{P}(\mathcal{E}_1^c|\mathcal{E}_0\cap\mathcal{E}_\infty)-\mathbb{P}(\mathcal{E}_2^c|\mathcal{E}_0\cap\mathcal{E}_\infty))
=(1-o(1))\kappa,
\end{align*}
as claimed.
As such, it suffices to show that $\mathbb{P}(\mathcal{F}\cap\mathcal{E}_{0}\cap \mathcal{E}_{1}\cap \mathcal{E}_{2}\cap \mathcal{E}_{\infty})$ is $o(1)$.
For this, it is convenient to consider the event $\mathcal{A}_{t_{0}} = \{e(I(t_{0}),S(t_{0})) > 0\}$ that the infection is still spreading at time $t_{0}$.
Since $\mathcal{E}_{1}$ is the event that the infection stays within $B(r)$ up to time $t_1$ (i.e., after $t_0$) and $\mathcal{E}_{\infty}$ is the event that the infection eventually escapes $B(r)$, it follows that $\mathcal{E}_{1}\cap \mathcal{E}_{\infty}\subseteq \mathcal{A}_{t_{0}}$.  
Since $\lambda\geq6\mu$ and $d\geq3$, we may select any $\epsilon\in(\frac{2\mu}{\mu+\lambda},1-\frac{2}{d})$; we will refine our choice later.
Defining $\mathcal{B}_{t_0}:=\{|R(t_{0})|< \epsilon |U(t_{0})|\}$, we then have
\begin{align}
\mathbb{P}(\mathcal{F}\cap\mathcal{E}_{0}\cap \mathcal{E}_{1}\cap \mathcal{E}_{2}\cap \mathcal{E}_{\infty})
\nonumber
&\leq \mathbb{P}(\mathcal{E}_{1}\cap \mathcal{E}_2\cap\mathcal{E}_{\infty}\cap\mathcal{B}_{t_0}^c)+\mathbb{P}(\mathcal{F}\cap\mathcal{E}_{0}\cap \mathcal{E}_{1}\cap \mathcal{E}_{2}\cap\mathcal{B}_{t_0})\\
\label{eq.two more terms to bound}
&\leq \mathbb{P}(\mathcal{E}_2\cap\mathcal{A}_{t_0}\cap\mathcal{B}_{t_0}^c)+\mathbb{P}(\mathcal{F}\cap\mathcal{E}_{0}\cap \mathcal{E}_{1}\cap \mathcal{E}_{2}\cap\mathcal{B}_{t_0}).
\end{align}
We bound the first term of \eqref{eq.two more terms to bound} by analyzing the underlying discrete time Markov chain:

\begin{lemma}
\label{lem.bound second term}
Select $\epsilon > \frac{2\mu}{\mu+\lambda}$.
Then $\mathbb{P}(\mathcal{E}_2\cap\mathcal{A}_{t_0}\cap\mathcal{B}_{t_0}^c)\leq \frac{e^{-c(\epsilon)k}}{1-e^{-c(\epsilon)}}$, where $c(\epsilon)>0$.
\end{lemma}

\begin{proof}
Let $M_n=(S_n,I_n,R_n)$ sequence of states of the SIR model.
For example, $M_0=(S(0),I(0),R(0))$, and $M_1=(S(T),I(T),R(T))$, where $T$ denotes the first transition time.
Almost surely, the end state of this process takes the form $(S,\emptyset,V(G)\setminus S)$.
Importantly, $\{M_n\}_{n\geq0}$ is a discrete time Markov chain in which, conditioned on $M_n$, it holds that $R_{n+1}$ strictly contains $R_n$ with probability
\[
\frac{\mu|I_{n}|}{\lambda e(I_{n},S_{n})+\mu|I_{n}|}.
\]
We will consider this process until the stopping time
\[
N:=\inf\{n:e(I_n,S_n)=0\}.
\]
Specifically, for each $n\in\{1,\ldots,N\}$, let $X_n$ indicate whether the $n$th transition recovers a vertex, i.e., $R_n$ strictly contains $R_{n-1}$.
For each $n>N$, let $X_n$ be Bernoulli with success probability $\frac{\mu}{\mu+\lambda}$, all of which are independent of each other and of $M_N$.
Put $\overline{M}_n:=M_{\min(n,N)}$.
Let $\mathbb{P}_n$ denote the (random) probability measure conditioned on the state history $\{\overline{M}_{j}\}_{j=0}^n$.
Notice that in the event $\{n<N\}$, we have the bound
\[
\mathbb{P}_n\{X_{n+1}=1\}
=\mathbb{P}\{X_{n+1}=1 |M_{n}\}
=\frac{\mu|I_{n}|}{\lambda e(I_{n},S_{n})+\mu|I_{n}|}
\leq\frac{\mu}{\mu +\lambda}
\]
almost surely.
Meanwhile, in the complementary event $\{n\geq N\}$, $X_{n+1}$ is Bernoulli with success probability $\frac{\lambda}{\mu +\lambda}$ and independent of $M_{n}$, and so
\[
\mathbb{P}_n\{X_{n+1}=1\}
=\mathbb{P}\{X_{n+1}=1 |M_{n}\}
=\mathbb{P}\{X_{n+1}=1\}
=\frac{\mu}{\mu +\lambda}
\]
almost surely.
Overall, $\mathbb{P}_n\{X_{n+1}=1\}\leq\frac{\mu}{\mu +\lambda}$ almost surely.
Next, let $J$ denote a set of positive integers $j_1<\cdots<j_m$.
Then the law of total probability gives
\[
\mathbb{P}\{X_j=1~\forall j\in J\}
=\mathbb{E}[\mathbb{P}_{j_m-1}\{X_j=1~\forall j\in J\}].
\]
Next, $\mathbb{P}_{j_m-1}\{X_j=1~\forall j\in J\setminus\{j_m\}\}\in\{0,1\}$, and so
\begin{align*}
\mathbb{P}_{j_m-1}\{X_j=1~\forall j\in J\}
&=\mathbb{P}_{j_m-1}\{X_j=1~\forall j\in J\setminus\{j_m\}\}\cdot\mathbb{P}_{j_m-1}\{X_{j_m}=1\}\\
&\leq \mathbb{P}_{j_m-1}\{X_j=1~\forall j\in J\setminus\{j_m\}\}\cdot\tfrac{\mu}{\mu +\lambda}.
\end{align*}
Take expectations of both sides and apply induction to get
\[
\mathbb{P}\{X_j=1~\forall j\in J\}
\leq \mathbb{P}\{X_j=1~\forall j\in J\setminus\{j_m\}\}\cdot\tfrac{\mu}{\mu +\lambda}
\leq
\cdots
\leq (\tfrac{\mu}{\mu +\lambda})^m.
\]
Next, let $N_0$ denote the number of transitions that have occurred by time $t_{0}$.  In the event $\mathcal{A}_{t_{0}}$, it holds that $N_0<N$, and so $|R(t_{0})|=\sum_{j=1}^{N_0} X_j$.
Also, we have $N_0\leq 2|U(t_0)|$.
As such,
\begin{align*}
\mathbb{P}(\{|U(t_0)|=u\}\cap\mathcal{A}_{t_{0}}\cap\{|R(t_{0})|\geq \epsilon u\})
&\leq \mathbb{P}\Big\{\sum_{j=1}^{2u} X_j\geq \epsilon u\Big\}\\
&\leq \sum_{\substack{J\subseteq[2u]\\|J|=\lceil \epsilon u\rceil}}\mathbb{P}\{X_j=1~\forall j\in J\}\\
&\leq\binom{2u}{\lceil \epsilon u\rceil}\Big(\frac{\mu}{\mu+\lambda}\Big)^{\lceil \epsilon u\rceil}\\
&\leq\operatorname{exp}\Big(-\lceil \epsilon u\rceil\Big(\log(\tfrac{\lambda}{\mu}+1)-\log(\tfrac{2u}{\lceil \epsilon u\rceil})\Big)\Big)\\
&\leq\operatorname{exp}\Big(- \epsilon u\Big(\log(\tfrac{\lambda}{\mu}+1)-\log(\tfrac{2}{ \epsilon})\Big)\Big)\\
&=: e^{-c(\epsilon) u},
\end{align*}
where $c(\epsilon)>0$ since $\epsilon> \frac{2\mu}{\mu+\lambda}$.
This then gives
\begin{align*}
\mathbb{P}(\mathcal{E}_2\cap\mathcal{A}_{t_0}\cap\mathcal{B}_{t_0}^c)
&=\mathbb{P}(\{|U(t_0)|\geq k\}\cap\mathcal{A}_{t_0}\cap\{|R(t_0)|\geq \epsilon|U(t_0)|\})\\
&=\sum_{u\geq k}\mathbb{P}(\{|U(t_0)|=u\}\cap\mathcal{A}_{t_{0}}\cap\{|R(t_{0})|\geq \epsilon u\})
\leq\sum_{u\geq k}e^{-c(\epsilon)u}
=\frac{e^{-c(\epsilon)k}}{1-e^{-c(\epsilon)}},
\end{align*}
as desired.
\end{proof}

Overall, Lemma~\ref{lem.bound second term} gives that the first term in \eqref{eq.two more terms to bound} is $o(1)$.
It remains to bound to second term in \eqref{eq.two more terms to bound}.
To do so, we restrict to the event $\mathcal{E}_{0}\cap \mathcal{E}_{1}\cap \mathcal{E}_{2}\cap\mathcal{B}_{t_0}$ and argue that $\mathcal{F}$ occurs with probability $o(1)$.
We adopt the notation from Lemma~\ref{lem.key lemma} and Algorithm~\ref{alg.main_algorithm}.
Since $B(r)$ is a tree, $A$ consists of the vertices in $U(t_{0})$ that have a neighbor in $S(t_{0})$, while $A_I=A\cap I(t_{0})$.
For every $a\in A\setminus A_I\subseteq R(t_0)$, since $a$ cannot infect $b(a)$, it holds that $Z_{a}=\tau$.
It follows that $Q=Q_I$ and $P=\frac{|A_I|}{|A|}P_I$, and so
\[
\frac{\hat\lambda}{\hat\lambda_I}
=\frac{P}{P_I}
=\frac{|A_I|}{|A|},
\qquad
\frac{\hat\mu}{\hat\mu_I}
=\frac{1-P}{1-P_I}
=\frac{1-\frac{|A_I|}{|A|}P_I}{1-P_I}
=1+\Big(1-\frac{|A_I|}{|A|}\Big)\cdot\frac{P_I}{1-P_I}.
\]
As such, $\frac{\hat\lambda}{\hat\lambda_I}\leq 1$ and $\frac{\hat\mu}{\hat\mu_I}\geq 1$.
To obtain bounds in the other directions, we require a lemma:

\begin{lemma}
\label{lem.bridge points}
Consider any $d$-regular graph $G$ and vertex subset $U\subseteq V(G)$ that induces a subtree of $G$.
At least $(1 - 2/d) |U|$ of the members of $U$ has a neighbor in $V(G)\setminus U$.
\end{lemma}

\begin{proof}
Let $W$ denote the vertices in $U$ with a neighbor in $U^c:=V(G)\setminus U$.
The number of edges in the tree induced by $U$ is $|U|-1$, while the total number of edges in $G$ incident to $U$ is given by
\[
\sum_{u\in U}\operatorname{deg}(u) - (|U|-1)
=(d-1)|U|+1.
\]
As such, the number of edges between $U$ and $U^c$ is
\[
e(U,U^c)
=((d-1)|U|+1) - (|U|-1)
=(d-2)|U| + 2.
\]
Pigeonhole then gives
\[
|W|
\geq \frac{e(U,U^c)}{d}
\geq \Big(1-\frac{2}{d}\Big)|U|.
\qedhere
\]
\end{proof}

Since $U(t_0)$ induces a subtree of $G$, Lemma~\ref{lem.bridge points} gives that $|A|\geq(1 - \frac{2}{d})|U(t_{0})|$, and so
\[
|A_I|
\geq|A|-|R(t_0)|
>|A|-\epsilon|U(t_0)|
\geq (1-\tfrac{\epsilon}{1 - \frac{2}{d}})|A|
\geq (1-3\epsilon)|A|
= (1-(1+\epsilon_0)\tfrac{6\mu}{\mu+\lambda})|A|,
\]
where the last two steps use the fact that $d\geq3$ and the choice $\epsilon=(1+\epsilon_0)\frac{2\mu}{\mu+\lambda}$ for some small $\epsilon_0>0$.
It follows that
\[
\frac{\hat\lambda}{\hat\lambda_I}
=\frac{|A_I|}{|A|}
=1-(1+\epsilon_0)\cdot\frac{6\mu}{\mu+\lambda}
\geq1-(1+\epsilon_0)\cdot\frac{6}{7}
\geq\frac{1}{8},
\]
where the last two steps use the facts that $\lambda\geq6\mu$ and $\epsilon_0$ is small.
Next, consider
\[
\tilde{k}
:=|A_I|
\geq|A|-\epsilon|U(t_0)|
\geq(1-\tfrac{2}{d}-\epsilon)|U(t_0)|
\geq(1-\tfrac{2}{d}-\epsilon)k,
\]
where $1-\tfrac{2}{d}-\epsilon>0$ by assumption.
Since $\tilde{k}$ and $\tau$ both increase with factors of $\log n$, Lemma~\ref{lem.key lemma} implies that $(\hat{\lambda}_I,\hat{\mu}_I)$ converges to $(\lambda,\mu)$ in probability.
As such,
\begin{align*}
\frac{\hat\mu}{\hat\mu_I}
=1+\Big(1-\frac{|A_I|}{|A|}\Big)\cdot\frac{P_I}{1-P_I}
&\leq1+(1+\epsilon_0)\cdot\frac{6\mu}{\mu+\lambda}\cdot\frac{\hat\lambda_I}{\hat\mu_I}\\
&=1+(1+\epsilon_0)\cdot\frac{6\mu}{\mu+\lambda}\cdot\frac{\lambda}{\mu}+o(1)
\leq 8,
\end{align*}
where the last step holds when $n$ is large.
The result then follows from the fact that $(\hat{\lambda}_I,\hat{\mu}_I)$ converges to $(\lambda,\mu)$ in probability.

\section{Discussion}

In this paper, we introduced a simple algorithm to estimate SIR parameters from early infections.
There are many interesting directions for future work.
First, we do not believe that Theorem~\ref{thm.parameter estimation locally tree like} captures the true performance of Algorithm~\ref{alg.main_algorithm}, especially in light of Figure~\ref{fig.phase_transitions}, and this warrants further investigation.
Next, it would be interesting to consider other types of estimators.
Notice that since Algorithm~\ref{alg.main_algorithm} explicitly makes use of certain properties of the underlying graph, it is clear why it fails for the complete graph.
Does the behavior of the maximum likelihood estimator have a similarly transparent dependence on the underlying graph?
We focused on locally tree-like graphs in part because there is a rich literature on SIR dynamics over infinite trees, but it would be interesting to analyze the performance of Algorithm~\ref{alg.main_algorithm} on other graph families.
Also, there is a multitude of compartmental models for epidemics with various choices of probability distributions for transition times between compartments.
Finally, one might consider alternative models for what data is available.
For example, to model asymptomatic infections, one might assume that a random fraction of infected vertices are not known to be infected.

\section*{Acknowledgments}

The authors thank Boris Alexeev for interesting discussions that inspired this work.
DGM was partially supported by AFOSR FA9550-18-1-0107 and
NSF DMS 1829955.

\end{document}